\documentclass[opre,nonblindrev]{informs3-rev}

\OneAndAHalfSpacedXI
\usepackage{endnotes}
\let\footnote=\endnote

\usepackage{natbib}
 \bibpunct[, ]{(}{)}{,}{a}{}{,}%
 \def\bibfont{\small}%
 \def\newblock{\ }%
\TheoremsNumberedThrough     % Preferred (Theorem 1, Lemma 1, Theorem 2)
\ECRepeatTheorems

\EquationsNumberedThrough    % Default: (1), (2), ...
\MANUSCRIPTNO{OPRE-2022-09-505.R4} % When the article is logged in and DOI assigned to it,
\usepackage{dsfont}
\usepackage{cases}

\def\N{\mathbb{N}}

\def\p{\mathbb{P}}
\def\E{\mathbb{E}}

\def\R{\mathbb{R}}

\def\X{\mathcal{X}}

\def\cov{\mathrm{cov}}

\def\X{\mathcal{X}}

\def\d{\mathrm{d}}
\DeclareMathOperator*{\esssup}{ess\text{-}sup}
\DeclareMathOperator*{\essinf}{ess\text{-}inf}

\newcommand{\VaR}{\mathrm{VaR}}
\newcommand{\ES}{\mathrm{ES}}

\newcommand{\RVaR}{\mathrm{RVaR}}

\def\id{\mathds{1}}

\renewcommand{\(}{\left(}
\renewcommand{\)}{\right)}

\renewcommand{\]}{\right]}
\def\laweq{\buildrel \d \over =}

\usepackage{enumerate} 
\usepackage{subcaption}
\usepackage{appendix}

\begin{document}
%%%%%%%%%%%%%%%%

%\setcounter{page}{1}
%\DOI{10.1287/xxxx.0000.0000}%

\RUNAUTHOR{Chen, Embrechts, and Wang}

\RUNTITLE{An unexpected stochastic dominance}

% \TITLE{Bundling Information Goods of Decreasing Value}
\TITLE{An unexpected stochastic dominance: Pareto distributions, dependence, and diversification}

\ARTICLEAUTHORS{%
\AUTHOR{Yuyu Chen}
\AFF{Department of Economics, University of Melbourne,  Australia, \EMAIL{yuyu.chen@unimelb.edu.au}} %, \URL{}}
\AUTHOR{Paul Embrechts}
\AFF{RiskLab, Department of Mathematics \& ETH Risk Center, ETH Zurich,  Switzerland, \EMAIL{embrechts@math.ethz.ch}}
\AUTHOR{Ruodu Wang}
\AFF{Department of Statistics and Actuarial Science, University of Waterloo,  Canada, \EMAIL{wang@uwaterloo.ca}}
} % end of the block

\ABSTRACT{%
We find the perhaps surprising inequality
that the weighted average of independent and identically distributed Pareto random variables with infinite mean is larger than one such random variable in the sense of first-order stochastic dominance.
This result holds for more general models including super-Pareto distributions, negative dependence, and triggering events, and yields superadditivity of the risk measure Value-at-Risk for these models.
}%

%\KEYWORDS{deterministic inventory theory; infinite linear programming duality;
%  existence of optimal policies; semi-Markov decision process; cyclic schedule}

%\HISTORY{This paper was
%first submitted on April 12, 1922 and has been with the authors for
%83 years for 65 revisions.}

\maketitle
%%%%%%%%%%%%%%%%%%%%%%%%%%%%%%%%%%%%%%%%%%%%%%%%%%%%%%%%%%%%%%%%%%%%%%

% Samples of sectioning (and labeling) in OPRE
% NOTE: (1) \section and \subsection do NOT end with a period
%       (2) \subsubsection and lower need end punctuation
%       (3) capitalization is as shown (title style).
%
%\section{Introduction.}\label{intro} %%1.
%\subsection{Duality and the Classical EOQ Problem.}\label{class-EOQ} %% 1.1.
%\subsection{Outline.}\label{outline1} %% 1.2.
%\subsubsection{Cyclic Schedules for the General Deterministic SMDP.}
%  \label{cyclic-schedules} %% 1.2.1
%\section{Problem Description.}\label{problemdescription} %% 2.

% Text of your paper here

\section{Introduction}\label{sec:1}
 Pareto distributions are arguably the most important class of heavy-tailed loss distributions, due to their connection to  regularly varying tails, Extreme Value Theory (EVT), and power laws in economics and social networks; see, e.g., \cite{EKM97}, \cite{DF06} and \cite{G09}. In quantitative risk management,  Pareto distributions are frequently used to model losses from catastrophes such as earthquakes, hurricanes, and wildfires; see, e.g., \cite{embrechts1999extreme}. They are  also widely used  in economics 
 for wealth distributions (e.g., \cite{T20}) and modeling the tails of financial asset losses and operational risks (e.g., \cite{MFE15}); applications of power laws in economics, finance, and insurance are treated in \cite{ibragimov2015heavy} and \cite{ibragimov2017heavy}. \cite{andriani2007beyond} listed over $80$ examples of power laws in diverse fields of applications. By the 
 Pickands-Balkema-de Haan Theorem (\cite{P75} and \cite{BD74}),
  generalized Pareto distributions are the only possible non-degenerate limiting distributions of the residual lifetime of random variables exceeding a high level. 
  
 In the realm of banking and insurance,
 distributions with infinite mean occur as a possible mathematical model in several contexts.
For instance,  catastrophic losses, operational losses, large insurance losses,  and financial returns from
technological innovations are often modelled by Pareto distributions without finite mean; 
see \cite{EKM97} in the context of extreme value theory,
  \cite{hofert2012statistical} on nuclear power accidents, and the more recent \cite{cheynel2022fraud} on modeling fraud.
%Appendix \ref{sec:rem1} collects some examples and related literature. 
In risk management, such infinite-mean models often lead to intriguing phenomena, such as the diversification disaster studied by \cite{IJW09,IJW11}.

 Stochastic dominance relations are an important tool in economic decision theory, which allows for the analysis of  risk preferences for a group of decision makers (\cite{HR69}).
They have been studied in the forms of first and second degrees (\cite{quirk1962admissibility}, \cite{HR69,HR71} and \cite{RS70}),   larger integer degrees (\cite{whitmore1970third} and \cite{CP96}), and fractional degrees (\cite{MSTW17} and \cite{HTZ20}), and they are widely applied in the expected utility and dual utility theory (\cite{Y87}),   behavioural decision models (\cite{CKS87}, \cite{BH06} and \cite{SZ08}), and risk measures (\cite{FS16}). See also \cite{levy1992stochastic,L16} for the wide applicability of stochastic dominance relations in decision making, and \cite{MS02} and \cite{SS07} for the mathematics of stochastic dominance.

 The strongest form of   commonly used  stochastic dominance relations is first-order stochastic dominance. 
For two random variables $X$ and $Y$ representing random losses, we say $X$ is smaller than $Y$ in \emph{first-order stochastic dominance}, denoted by $X\le_{\rm st}Y$, if $\p(X\le x)\ge \p(Y\le x)$ for all $x\in\R$.  
The relation $X\le_{\rm st}Y$ means that all decision makers with an increasing utility function will prefer  loss $X$ to  loss $Y$ if their expected utilities are finite. 
 In this paper, all terms like ``increasing" and ``decreasing" are in the non-strict sense.

For iid random variables $X_{1},\dots,X_{n}$ following a Pareto distribution with infinite mean
and weights $\theta_{1},\dots,\theta_n\ge 0$ with $\sum_{i=1}^{n}\theta_i=1$, one consequence of our main result, Theorem \ref{thm:1}, is the stochastic dominance relation
\begin{equation}\label{eq:*}
X_{1}\le_{\rm st}\theta_{1}X_{1}+\dots+\theta_{n}X_{n},
\end{equation}  
and \eqref{eq:*} is not an equality except for the trivial case that only one of $\theta_1,\dots,\theta_n$ is non-zero. 
As far as we are aware,   \eqref{eq:*} is not known in the literature, even in the case that $\theta_1,\dots,\theta_n$ are equal (i.e., they are $1/n$).
%First-order stochastic dominance is arguably the strongest form of robust risk comparison: 
%For all  monotone decision makers (with precise definition in Section \ref{sec:4}) such as the ones with an increasing utility function, a weighted average of iid Pareto losses less preferred compared to one of its component.

To appreciate the nature of \eqref{eq:*}, we first recall that for any identically distributed random variables $X_1,\dots,X_n$ with finite mean, regardless of their distribution or dependence structure, for $\theta_1,\dots,\theta_n>0$ with $\sum_{i=1}^n \theta_n=1$,  \eqref{eq:*}
can only hold if $X_1=\dots =X_n$ (almost surely),
in which case we have the trivial equality $X_1= \theta_{1}X_{1}+\dots+\theta_{n}X_{n} $; see Proposition \ref{prop:equality}. Therefore, the assumption of infinite mean is very important for \eqref{eq:*} to hold.

It is somewhat surprising that, for infinite-mean Pareto losses, inequality \eqref{eq:*} holds for a very  strong form of risk comparison: 
For every decision maker with a risk preference favouring less loss over more loss and \emph{well defined} for losses in \eqref{eq:*},
a diversified portfolio of such iid Pareto losses is less preferred to a non-diversified one.
Flipping the sign, diversification is preferred if the Pareto random variables are treated as profits or gains from, for instance, research and development. We call such a stochastic dominance ``unexpected" for both its surprising nature and the infinite expectations involved.

The infinite mean assumption comes with a caveat: for a risk-averse expected utility decision maker, due to the concavity of the utility function, \eqref{eq:*} does not imply a preference for non-diversification, because losses in both sides of \eqref{eq:*} have $-\infty$ expected utility.  
%See Section \ref{sec:4} for related discussions.
It does, however,    give \emph{superadditivity} of the regulatory risk measure Value-at-Risk (VaR) in banking and insurance sectors; that is, the weighted average of super-Pareto losses gives a larger  VaR than the weighted average of VaRs of individual super-Pareto losses.
%that given by an individual super-Pareto loss. 
Different from the literature on VaR superadditivity for regularly varying distributions (e.g., \cite{ELW09} and \cite{MFE15}), the superadditivity of VaR implied by \eqref{eq:*} holds for all probability levels, and this is not in an asymptotic sense.

Observations similar to \eqref{eq:*}, although with less generality, occur in the literature in different forms. 
\cite{S67} mentioned that having an infinite mean in portfolio diversification may lead to a worse distribution; see also   \citet[p.~271]{FM72} and \cite{M72}. 
Inequality \eqref{eq:*} for $n=2$ and the Pareto tail parameter $\alpha =1/2$  (see Section \ref{sec:2} for the parametrization) has  an explicit formula in Example 7 of \cite{embrechts2002correlation}. Simple numerical examples are provided by \citet[Figure 5.2]{EP10} and \citet[Table 2]{bauer2016marginal}.
% A numerical example for $n=3$ and $\alpha=1$ is provided by \citet[Figure 5.2]{EP10}. 
%(see Section \ref{sec:2} for the parameterization of Pareto distributions), 
%Example 7 of \cite{embrechts2002correlation} showed that,
%\begin{equation}\label{eq:1}
%X\le_{\rm st}\frac{X+Y}{2}.
%\end{equation} We generalize \eqref{eq:1} to $n$ Pareto losses with any parameter $\alpha \le 1$.   
\cite{ibragimov2005new} showed that \eqref{eq:*} holds for iid positive one-sided stable random variables with infinite mean.
Another relevant result  of \cite{ibragimov2009portfolio} is  that for iid  random variables $Z_1,\dots,Z_n$ which follow a convolution of symmetric stable distributions without finite mean, $\p(\theta_1Z_1+\dots+\theta_nZ_n\le x)\le \p(Z_1\le x)$ for $x>0$ but the opposite holds for $x<0$ (and hence   first-order stochastic dominance does not hold). 
%\footnote{This means that $\theta_1Z_1+\dots+\theta_nZ_n$ is ``more spread out" than $Z_1$. This notion is closer to second-order stochastic dominance, which captures mean-preserving spreads (although the mean is infinite); see \cite{IW07}.}).
The symmetry of distributions is essential for this inequality, and    $Z_1,\dots,Z_n$ can take negative values, unlike Pareto losses, which are positive,  skewed and more suitable for the modeling of extreme losses.

  %The inequality \eqref{eq:*} implies that the diversification of iid Pareto losses without finite mean leads to a larger probability that the loss exceeds \emph{any threshold}, which we call a \emph{diversification penalty}.
   % In risk management, inequality \eqref{eq:*}    yields  \emph{superadditivity} of the regulatory risk measure Value-at-Risk (VaR) in banking and insurance sectors; that is, the weighted average of Pareto losses without finite mean gives a larger  VaR than that given by an individual Pareto loss. Different from the literature on VaR superadditivity for regularly varying distributions (e.g., \cite{ELW09} and \cite{MFE15}), the superadditivity of VaR implied by \eqref{eq:*} holds for all probability levels, and this is not just in some asymptotic sense.

In Section \ref{sec:2}, we first introduce super-Pareto distributions,
a class of infinite-mean distributions, 
and  weak negative association, 
a  notion of dependence weaker than negative association (\cite{AS81} and \cite{JP83}). The class of super-Pareto distributions 
includes all infinite-mean Pareto distributions. 
Our main result, 
Theorem \ref{thm:1}, shows  that \eqref{eq:*} holds if   $X_1,\dots,X_n$ are weakly negatively associated super-Pareto random variables,
and also in case they are  triggered by some events. Some discussions on this result are provided after its proof.
%Generalizations of the inequality \eqref{eq:*} 
%  are obtained in   Section \ref{sec:2}. All of these generalizations consider a broad class of heavy-tailed distributions called super-Pareto distributions, which includes the class of infinite-mean Pareto distributions.
% In particular, Theorem \ref{thm:1} shows  that \eqref{eq:*} holds for weakly negatively associated super-Pareto random variables, each caused by a triggering event. 
% Propositions \ref{prop:convolution} and \ref{prop:tail} deal with random variables that are convolutions of super-Pareto random variables and random variables that are super-Pareto only in the tail region, respectively. Proposition \ref{prop:mixture} handles super-Pareto random variables that are dependent through a mixture of some copulas.
%  Section \ref{app:insurance} studies a classic insurance model with random summation and weights.  
%We discuss in Section \ref{sec:4} decision models for which Theorem \ref{thm:1} can be applied; they only need to satisfy (i) choice under risk and (ii) less loss is better. 
%First,  
%inequality \eqref{eq:*} naturally yields strict preference for non-diversification in many decision models. 
 %Third, although  \eqref{eq:*}  never holds for non-degenerate random variables with finite mean, Theorem \ref{prop:bounded} shows that a similar diversification penalty exists for truncated super-Pareto losses, as long as the thresholds are high enough.
%The class of distributions satisfying \eqref{eq:*} beyond Theorem \ref{thm:1} is  discussed in Section \ref{sec:extension}.
 Section \ref{sec:7} concludes. %the paper.  
 
%The e-companion has five sections supporting the paper.
%Proofs of all results, except for  Theorem \ref{thm:1}, are in Section \ref{app:proof}.
%  Some generalizations of our main result to other models 
%  are presented in Section \ref{app:gen}. Additional results on the risk management decision of a single agent are in Section \ref{app:A} and an equilibrium analysis in some settings of a risk exchange economy is  in Section \ref{app:r1}. Some numerical examples are presented in Section \ref{sec:6}.
  % Section \ref{app:r1} contains  generalizations of the risk exchange market studied in Section \ref{sec:5}. 

   We fix some notation. Throughout, random variables are defined on an atomless probability space $(\Omega,\mathcal F,\p)$. Denote by $\mathbb N$ the set of all positive integers and $\R_+$ the set of non-negative real numbers. For $n\in \N$, let $[n]=\{1,\dots,n\}$.  Denote by $\Delta_n$   the standard simplex, that is, $\Delta_n=\{(\theta_1,\dots,\theta_n)\in [0,1]^n:  \sum_{i=1}^n \theta_i=1\}$. 
For $x,y\in \R$, write $x\wedge y= \min\{ x,y\}$, $x\vee y= \max\{ x,y\}$, and $x_+=\max\{x,0\}$.  We write $\mathbf X\laweq \mathbf Y$ if $\mathbf X$ and $\mathbf Y$ have the same distribution. We always assume $n\ge 2$. Equalities and inequalities are
interpreted component-wise when applied to vectors.
 For any random variable $X$, its essential infimum is given by $z_X=\inf\{z\in \R: \p(X>z)>0\}$.

%\end{remark}

%\subsection*{Notation}

\section{Stochastic dominance for super-Pareto risks}\label{sec:2}

\subsection{Super-Pareto distributions and weak negative association}

 We first introduce the Pareto distribution and the super-Pareto distribution. %, used throughout the paper.
% For a random variable $X$, denote by $q_X:(0,1)\to \R$  
% its quantile function, defined as
% $$
% q_X(p)= \inf\{x\in \R:\p(X\le x) \ge p\}. 
% $$  
% Moreover, let $q_X(0) = \lim_{p\downarrow 0} q_X(p)$, which is the essential infimum of $X$.
A common parameterization of Pareto distributions is given by, for $\theta,\alpha>0$, the cumulative distribution function
 \begin{align*}%\label{eq:PD}
  P_{\alpha,\theta}(x) = 1 -\left(\frac{\theta}{x}\right)^{\alpha},~~x\ge \theta.
  \end{align*}
As $\theta$ is a scale parameter, it suffices to study   $ P_{\alpha,1} $, which we write   as $ \mathrm {Pareto}(\alpha)$.
 % For $X\sim P_{\alpha,\theta }$,   $q_X(p)=\theta (1-p)^{-1/\alpha}$ for $p\in [0,1)$. 
 The mean of $  P_{\alpha,\theta}$ is infinite if and only if the tail parameter $\alpha $ is in $ (0,1]$. We say that the $  P_{\alpha,\theta}$ distribution is \emph{extremely heavy-tailed} if $\alpha\le 1$.
 %, and it is \emph{moderately heavy-tailed} if $\alpha >1$. 

\begin{definition}\label{def:1}
A random variable   $X$ is \emph{super-Pareto} (or has a super-Pareto distribution)  if $X\laweq f(Y)$ for some increasing, 
convex, and non-constant function $f$ and $Y\sim \mathrm{Pareto}(1)$. Moreover, $X$ is \emph{regular} if 
$f(0)=0$ and $f(1)>0$.  
\end{definition} 
All extremely heavy-tailed Pareto distributions are super-Pareto and regular. By definition, the super-Pareto property is preserved under increasing, convex, and non-constant transforms, including location-scale transforms. 
Intuitively, increasing convex transforms, such as $x\mapsto x^\beta$ for $\beta>1$   and  $x\mapsto \exp(x)$, generally make the tail of a random variable heavier. 
Thus, super-Pareto risks have  heavier tails than (or equivalent to) Pareto($1$) risks.
It is straightforward to check that any super-Pareto random variable has infinite mean.

Some examples of the super-Pareto family include the generalized Pareto distribution when $\xi\ge1$, specified  by 
 $$
\p(X\le x)=1-\(1+\xi\frac{x}{\beta}\)^{-1/\xi},~~~ x\ge 0,
 $$
where   $\beta>0$,
and the Burr distribution when $\alpha,\tau\in (0,1]$, specified  by 
$$\p(X\le x)=1-\(\frac{1}{x^\tau+1}\)^\alpha,~~~ x\ge 0.$$
 %Let $X\sim \mathrm{Pareto}(\alpha)$, $\alpha>0$. 
%Then $(X-1)^{1/\tau}$ follows a Burr distribution. Hence,  a  Burr distribution is super-Pareto if $\alpha,\tau\le1$.
Special cases of the Burr family are the paralogistic ($\alpha=\tau$) and the log-logistic ($\alpha=1$) distributions.

The next proposition gives an equivalent formulation for super-Pareto distributions, which will become useful in proving some results. 
 
\begin{proposition}\label{prop:R2-1} 
A random variable $X$ with essential infimum $z_X\in \R$ is   {super-Pareto} 
 if and only if the function 
 $g: x\mapsto 1/\p(X>x)
 $ is strictly increasing and concave 
 on $[z_X,\infty)$. 
If further $X$ is  {regular}, then 
$z_X> 0$  
and
$g(x) \le  x /z_X$ for $x \ge z_X$.  
\end{proposition}

%The connection between Pareto and super-Pareto distributions will be made more transparent by an alternative formulation in the next proposition, which can be checked directly.
% The next proposition gives an alternative formulation of super-Pareto distributions, which can be checked directly.
%\begin{proposition}\label{prop:R2-1}
%A random variable   $X$ is super-Pareto if and only if $X\laweq f(Y)$ for some increasing, 
%convex, and non-constant function $f$ and $Y\sim \mathrm{Pareto}(1)$. Any super-Pareto random variable has infinite mean.
%\end{proposition} 

%We say that a  random variable $X$ is  \emph{super-Pareto}
%if $t\mapsto q_X(1-1/t)$ is convex and non-constant on $(1,\infty)$. In this case, we also say that $X$ has a super-Pareto distribution.
%Note that for $X\sim   P_{\alpha,\theta}$,
%$q_X(1-1/t)= \theta (1-(1-1/t))^{-1/\alpha}= \theta t^{1/\alpha}$ for $t\in  (1,\infty)$.
%Hence, all extremely heavy-tailed Pareto distributions are super-Pareto.
%We can check that $X$ is super-Pareto if and only if $X$ is identically distributed as $f(Y)$ for some increasing, 
%convex, and non-constant function $f$ and $Y\sim \mathrm{Pareto}(1)$. We assume that the associated $f$ functions for all super-Pareto random variables considered in this paper satisfy $f(0)=0$. Intuitively, convex transformations, such as $x\mapsto x^\beta$ for $\beta>1$,  $x\mapsto \exp(x)$, and $x\mapsto (x-m)_+$ for some $m\in \R$, generally make the tail of a random variable heavier. 
%Thus, super-Pareto risks have a heavier tail than Pareto($1$) risks. 

Next, we introduce a new notion of negative dependence.  

\begin{definition}

A set $S\subseteq \R^{k}$, $k\in \N$ is  \emph{decreasing} 
if $\mathbf x\in S$ implies $\mathbf y\in S$ for all $\mathbf y\le \mathbf x$. Random variables $X_1,\dots,X_n$ are \emph{weakly negatively associated} if  
 for any $i\in[n]$,  decreasing set $S  \subseteq  \R^{n-1}$, and $x\in \R$ with $\p(X_i\le x)>0$,  
\begin{equation}\label{eq:WNA}
 \p(\mathbf X_{-i} \in S  \mid  X_i\le  x) \le \p(\mathbf X_{-i}\in S),
 \end{equation} 
 where $\mathbf X_{-i}=(X_1,\dots,X_{i-1}, X_{i+1},\dots,X_n)$.
\end{definition}

Weak negative association includes independence as a special case. 
Intuitively, \eqref{eq:WNA} means that knowing $X_i$ is small implies that $\mathbf X_{-i}$ is less likely to be small, thus a concept of negative dependence. 
Moreover, \eqref{eq:WNA} implies, by flipping signs,
\begin{equation}\label{eq:inset}
\p(\mathbf X_{-i} > \mathbf x  \mid  X_i>  x) \le \p(\mathbf X_{-i}> \mathbf x)
\end{equation}
for   $i\in[n]$, $x\in \R$, and  
 $\mathbf x \in \R^{n-1}$. 
 % Interpreting $\mathbf X$ as a vector of profits, inequality \eqref{eq:inset} reflects on the  ``winner-takes-all'' phenomenon relevant in innovation races, that is, a single firm will gain a significant amount of profits over its competitors (see e.g., \cite{shapiro1999information}).  

Weak negative association is weaker than the popular notion of negative association (\cite{AS81}  and \cite{JP83}),  hence the name.
Examples satisfying negative association, such as normal distributions with non-positive correlations, are studied by \cite{JP83}. 
It is also implied by  negative regression dependence (\cite{L66} and \cite{block1985concept}), and it implies negative orthant dependence (\cite{block1982some}); see Remark \ref{rem:ND} for more details on these notions of dependence.

In most results, we consider weakly negatively associated and identically distributed (WNAID) super-Pareto random variables. This setting includes the iid Pareto($\alpha$)  model for $\alpha \in (0,1].$

\subsection{The main result}
For random variables $X$ and $Y$, we write $X<_{\rm st}Y$ if $\p(X>x)<\p(Y>x)$ for all $x\in \R$ satisfying $\p(X>x)\in (0,1)$, and this will be referred to as strict stochastic dominance. 
Note that this condition is stronger than a different notion of strict stochastic dominance defined by $X\le _{\rm st} Y$ and $X\not \ge_{\rm st} Y$.
The following theorem is  our main result. 
\begin{theorem}\label{thm:1}
Suppose that
 $X_1,\dots,X_n$ are WNAID super-Pareto,
$(\theta_1,\dots,\theta_n)\in\Delta_n$, and $X\laweq X_1$.
 \begin{enumerate}[(i)]
\item  Stochastic dominance holds:
\begin{equation} \label{eq:maineq1} 
X \le_{\rm st}\sum_{i=1}^n\theta_{i}X_{i} ,
\end{equation}
and strict stochastic dominance $X <_{\rm st}\sum_{i=1}^n\theta_{i}X_{i}$ holds if $\theta_i>0$ for at least two $i\in [n]$.
\item  If $X$ is regular, then for any events  $A_1,\dots,A_n$  independent of $(X_1,\dots,X_n)$ and  event $A$ independent of $X$  satisfying $\p(A)=\sum_{i=1}^n \theta_i\p(A_i)$,  we have
\begin{equation} \label{eq:maineq2} 
X \id_A\le_{\rm st}\sum_{i=1}^n\theta_{i}X_{i}\id_{A_i}.
\end{equation}
\end{enumerate}
\end{theorem}  

We will say that a \emph{diversification penalty} exists if \eqref{eq:maineq1} or \eqref{eq:maineq2} holds.  
%which is naturally interpreted as that having exposures in multiple super-Pareto losses is worse than having just one  super-Pareto loss of the same total exposure.   
To interpret Theorem \ref{thm:1}, the left-hand side of \eqref{eq:maineq1}  can be regarded as the loss of an agent who keeps their own risk, and the right-hand side is the loss of an agent who shares risks with other agents.    
By pooling among super-Pareto losses,  agents expect to suffer less loss when their own loss occurs. However, every agent in the pool will have a higher frequency of bearing losses. Hence, diversification has two competing effects on the loss portfolio: It increases the frequency of losses and decreases the sizes of individual losses.  The above results show that the combined effects of diversification of super-Pareto losses lead to a higher default probability at any capital reserve level, that is, $\p(\sum_{i=1}^n\theta_{i}X_{i}>x)> \p(X>x)$ for all $x>1$. 

%\begin{theorem}\label{thm:1}
%Suppose that
% $X_1,\dots,X_n$ are super-Pareto, identically distributed, and weakly negatively associated. 
% Let $A_1,\dots,A_n$ be any events independent of $(X,X_1,\dots,X_n)$.
%%Let $(X,X_1,\dots,X_n)$ satisfy the super-Pareto model and $A_1,\dots,A_n$ be any events independent of $(X_1,\dots,X_n)$.
%%Let $(Y_1,\dots,Y_n)$ be the catastrophe risk model associated with the super-Pareto model $(X_1,\dots,X_n)$.
%For $(\theta_1,\dots,\theta_n)\in\Delta_n$,  we have
%\begin{equation}
%\label{eq:maineq} 
%X \id_A\le_{\rm st}\sum_{i=1}^n\theta_{i}X_{i}\id_{A_i},\end{equation}
%where $X=X_1$ and $A$ is independent of $X$  satisfying $\p(A)=\sum_{i=1}^n \theta_i\p(A_i)$.
%%\begin{equation}
%%\label{eq:maineq-new} 
%%Y\le_{\rm st}\sum_{i=1}^n\theta_{i}Y_{i}.\end{equation}
%Moreover, for $t>q_X(0)$, $\p\(\sum_{i=1}^n\theta_{i}X_i>  t\)> \p\(X> t\)$ if $\theta_i>0$ for at least two $i\in [n]$.
%\end{theorem}  

The model in Theorem \ref{thm:1} (ii) 
reflects 
that catastrophic losses are large losses but usually occur with very small probabilities.
Such losses are usually modelled by $X\id_A$, where $X$ is a positive random loss, and $A$ is an event with a small probability, indicating the occurrence of a catastrophe, such as an earthquake or a flood.
% ; see e.g., \cite{CDM03}.} 
Note that in (ii), the events $A_1,\dots,A_n$ are arbitrary, meaning that  $X_1,\dots,X_n$ may be caused by either the same or different triggering events. 
In particular, if $A_1,\dots,A_n$ are the same, then $A$ can also be chosen as the same event, and in this case \eqref{eq:maineq2} follows from \eqref{eq:maineq1}.
More generally,  the condition $\p(A)=\sum_{i=1}^n \theta_i\p(A_i)$ makes it fair to compare the two sides of \eqref{eq:maineq2}; for instance, if $X$ has a finite mean instead of being super-Pareto, then both sides of \eqref{eq:maineq2} would have the same mean.
 % Consider an agent who faces WNAID super-Pareto losses $X_1,\dots,X_n$, possibly triggered by catastrophic events $A_1,\dots,A_n$ that are independent of $(X_1,\dots,X_n)$. Since catastrophic losses are large losses but usually occur with very small probabilities, it is practical to model  the agent's  loss portfolio as $(X_1\id_{A_1},\dots,X_n\id_{A_n})$ such that $X_i\id_{A_i}$ is super-Pareto conditional on $A_i$, $i\in[n]$. The events  $A_1,\dots,A_n$ are arbitrary, meaning that  $X_1,\dots,X_n$ may be caused by either the same or different triggering events. Let $(\theta_1,\dots,\theta_n)\in \Delta_n$ be the agent's exposure vector. The total loss of the agent can then be written as $\theta_1X_1\id_{A_1}+\dots+\theta_nX_n\id_{A_n}$. If $\p(A_1)=\dots=\p(A_n)=1$,  the total loss of the agent is simply $\theta_1X_1+\dots+\theta_nX_n$. Hence,  Theorem \ref{thm:1} presents  strong stochastic dominance relations for the agent's total loss. 
  Although our setting mainly concerns the losses of an agent,  it  is also  applicable to the setting of investment decisions. For instance, $X_1,\dots,X_n$ can represent profits from technology innovations modelled by $A_1,\dots,A_n$; negatively dependent profits may arise  in innovation races (see e.g., \cite{shapiro1999information}).  %Nevertheless, we will mainly consider $X_1,\dots,X_n$ as losses throughout this paper and study the implications of Theorem \ref{thm:1} in risk management.

An immediate consequence of Theorem \ref{thm:1}  is that  if super-Pareto random variables $X_1,\dots,X_n$ are independent and comparable in first-order stochastic dominance, for $(\theta_1,\dots,\theta_n)\in\Delta_n$, we have $ 
X_{i^*}\le_{\rm st}\sum_{i=1}^n\theta_{i}X_{i}$ if $X_{i^*}\le_{\rm st}X_i$ for all $i\in[n]$.

To better understand   the result in Theorem \ref{thm:1}, 
we stress that \eqref{eq:maineq1} and \eqref{eq:maineq2} cannot be expected if $X_1,\dots,X_n$ have finite mean, regardless of their dependence structure, as summarized in the following proposition. 
\begin{proposition}
\label{prop:equality}
For $\theta_1,\dots,\theta_n>0$ with $\sum_{i=1}^n \theta_n=1$ and identically distributed random variables $X,X_1,\dots,X_n$ with finite mean and any dependence structure, 
\eqref{eq:maineq1} holds if and only if $X_1=\dots=X_n$ almost surely.
\end{proposition}
 Proposition \ref{prop:equality} implies, in particular, that
\eqref{eq:maineq1} never holds for WNAID non-degenerate random variables $X,X_1,\dots,X_n$ with finite mean.
%As such, Theorem \ref{thm:1} yields a clear and elegant methodological distinction between the two modeling environments; 
%the difference between finite-mean and infinite-mean acts as a kind of phase-type transition concerning diversification.
Even if   $X,X_1,\dots,X_n$  have an infinite mean, 
we are not aware of any other distributions for which \eqref{eq:maineq1} and \eqref{eq:maineq2} hold other than the super-Pareto distributions studied in this paper.  
% all built on the basis of Theorem \ref{thm:1}. 

Risk measures are popular tools to quantify  the risk of a financial portfolio; see \cite{ADEH99} and \cite{FS16}. %Although risk measures are usually not used for preference models but for risk assessment, most of them induce a relation $\succeq$ that satisfies (i)--(iii), so that Theorem \ref{thm:1} can be applied. 
Value-at-Risk (VaR) is one of the   most widely used classes of risk measures in  financial regulation.
  For a random variable  $X$ with distribution function $F_X$, VaR at level  $p\in (0,1)$ is defined as the 
  (left) quantile of $X$ at $p$, that is,
%\begin{equation}\label{VaR}
$$
  \VaR_{p}(X)=F^{-1}_X(p)=\inf\{t\in\R:F_X(t)\geq p\}.
$$
% \end{equation}
 For any random variable $X$, $\VaR_p(X)$ is always finite, making it suitable for assessing losses with infinite mean.   Moreover, for two random variables $X$ and $Y$, $X\le_{\rm st} Y$ (resp.~$X<_{\rm st} Y$) if and only if $\VaR_p(X)\le \VaR_p(Y)$ (resp.~$\VaR_p(X)< \VaR_p(Y)$) for all $p\in (0,1)$.   % Many other commonly used risk measures can be written as integrals of VaR (see  Section \ref{app:A}).
  % assuming that 
% $\theta_i>0$ for at least two $i\in [n]$  and 
Noting the equality $\VaR_p(X_1)= \sum_{i=1}^n\theta_i\VaR_{p}(X_i)$ for iid random variables,    Theorem \ref{thm:1} gives superadditivity:
\begin{equation}\label{eq:result-VaR-2}\VaR_{p}\left( \sum_{i=1}^n \theta_{i}X_{i}\right) >  \sum_{i=1}^n\theta_i\VaR_{p}(X_i).\end{equation} 
Inequality \eqref{eq:result-VaR-2} and its non-strict version can be intuitively seen as diversification penalty for $\VaR_p$.
%and they  also hold for truncated super-Pareto distributions that are truncated at a very high level. 
  
\subsection{Proof of Theorem \ref{thm:1}}\label{sec:proof}

%Note that $X$ is super-Pareto if and only if $X\laweq f(Y)$  for some increasing,
%convex, and non-constant function $f$ and $Y\sim \mathrm{Pareto}(1)$.
% Let $X_1,\dots,X_n$ be weakly negatively associated random variables, $X\laweq X_1$, and $(\theta_1,\dots,\theta_n)\in \Delta_n$.
% The proofs of \eqref{eq:maineq1} and \eqref{eq:maineq2} are given below separately.
% In this section, we prove Theorem \ref{thm:1}.
% A technical lemma will be used in the proof, which is presented at the end of the section. 
 
We first show in step (a) that the result in (i) holds for two independent Pareto($1$) losses.  An induction argument yields the result for the sums of $n$ such losses.
We then extend the result to WNAID 
Pareto losses in step (b), and further to
super-Pareto losses in step (c). 
Finally, step (d) proves the result in (ii)   by applying   (i) and  analysis on the combination of indicator functions, where the regularity assumption of $X$ is needed.

 (a) Let $Y,Y_{1},\dots,Y_{n}\sim \mathrm{Pareto}(1)$ such that $Y_{1},\dots,Y_{n}$ are independent and $\Delta^+_n=\Delta_n\cap (0,1)^n$.  It suffices to show 
\begin{equation}\label{eq:IN1}
\p(\theta_1Y_1+\dots+\theta_nY_n\ge x)>\frac{1}{x} \mbox{~for all $x\in (1,\infty)$ and $(\theta_1,\dots,\theta_n)\in \Delta_n^+$}.
\end{equation}
Let $\delta=(x-1+\theta_n)/\theta_n$.
We will use the following fact: For $(y_1,\dots,y_n)\in(1,\infty)^n$, if $y_n\ge \delta$, then 
$$\theta_1y_1+\dots+\theta_ny_n\ge (\theta_1+\dots+\theta_{n-1})+\theta_n\delta=1-\theta_n+x-1+\theta_n=x.$$
We first show the case of $n = 2$. 
For all $x\in (1,\infty)$ and $(\theta_1,\theta_2)\in \Delta_2^+$ (see Figure \ref{fig:r3-1}), 
\begin{figure}
\begin{center}\includegraphics[height=6cm]{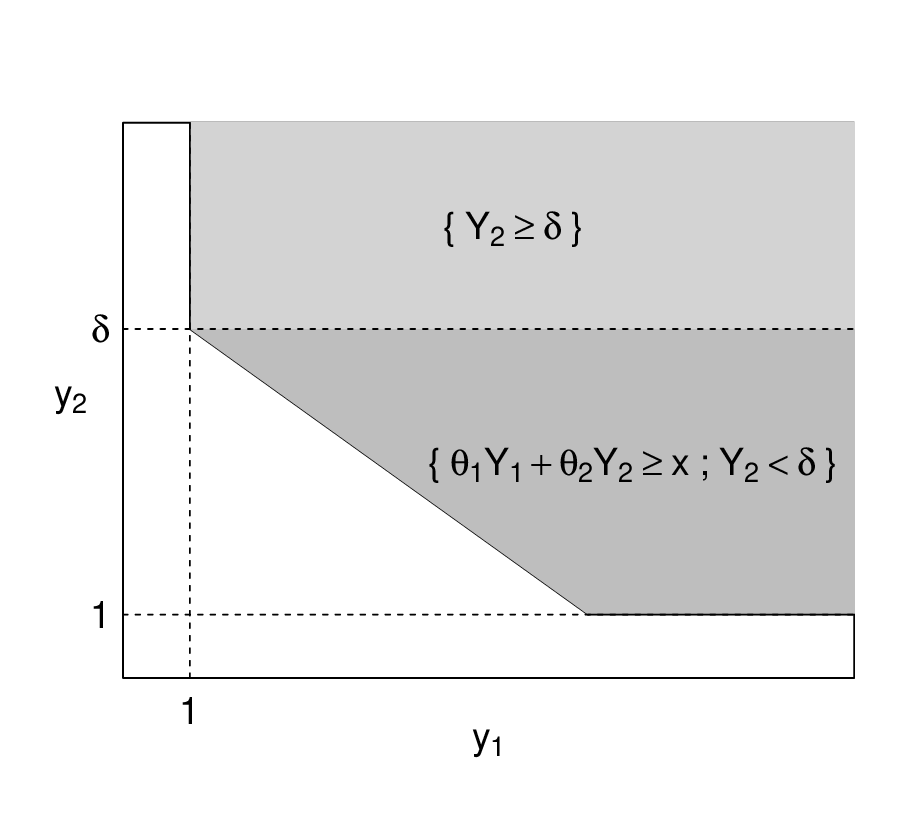}
\caption{An illustration of $\{\theta_1Y_1+\theta_2Y_2\ge x\} = \{ Y_2\ge \delta\}\cup \{\theta_1Y_1+\theta_2Y_2\ge x;Y_2< \delta\}$, where  $x=5$, $\theta_1 = 2/5$, $\theta_2 =3/5$, and $\delta=23/3$}
\label{fig:r3-1}
\end{center}
\end{figure}
\begin{align}
\p(\theta_1Y_1+\theta_2Y_2\ge x)&=\p(Y_2\ge \delta)+\p(\theta_1Y_1+\theta_2Y_2\ge x;Y_2< \delta) \nonumber \\ 
&\ge \p(Y_2\ge \delta)+\p(\theta_1Y_1+\theta_2\ge x;Y_2< \delta)\nonumber \\
&=\frac{1}{\delta}+\frac{\theta_1}{x-\theta_2}\(1-\frac{1}{\delta}\) \label{eq:proofIN1}  \\ 
&>\frac{1}{\delta}+\frac{\theta_1}{x}\(1-\frac{1}{\delta}\)  \nonumber  \\ 
&= \frac{1}{x\delta }(x+ \delta-\theta_2 \delta  -1+\theta_2)=\frac{1}{x}.\nonumber
\end{align}
Next, let $n>2$ and assume \eqref{eq:IN1} holds true for the case of $n-1$. Then, for all $x\in (1,\infty)$ and $(\theta_1,\dots,\theta_n)\in \Delta_n^+$,
\begin{align}
\p(\theta_1Y_1+\dots+\theta_nY_n\ge x)&=\p(Y_n\ge \delta)+\p(\theta_1Y_1+\dots+\theta_nY_n\ge x;Y_n< \delta) \nonumber \\ 
&\ge \p(Y_n\ge \delta)+\p(\theta_1Y_1+\dots+\theta_{n-1}Y_{n-1}\ge x-\theta_n;Y_n< \delta)\nonumber \\
&=\p(Y_n\ge \delta)+\p\(\frac{\theta_1Y_1+\dots+\theta_{n-1}Y_{n-1}}{1-\theta_n}\ge \frac{x-\theta_n}{1-\theta_n}\)\p(Y_n< \delta) \label{eq:proofIN2}\\
&\ge \frac{1}{\delta}+ \frac{1-\theta_n}{x-\theta_n}\(1-\frac{1}{\delta}\)  \nonumber \\
&> \frac{1}{\delta}+ \frac{1-\theta_n}{x}\(1-\frac{1}{\delta}\) =\frac{1}{x}. \nonumber 
\end{align}
By induction,   \eqref{eq:maineq1} holds for  $n$ independent $\mathrm{Pareto}(1)$ random variables.

(b)  Since $\{(y_1,\dots,y_n):\theta_1y_1+\dots+\theta_ny_n<t\}$ is a decreasing set, if $Y_1,\dots,Y_n$ are weakly negatively associated, for any $t>1$ and $n\ge 2$,
\begin{align*}
\p(\theta_1Y_1+\dots+\theta_{n-1}Y_{n-1}\ge t;Y_n<\delta)&=\p(Y_n<\delta)-\p(\theta_1Y_1+\dots+\theta_{n-1}Y_{n-1}< t;Y_n<\delta)\\
&\ge \p(Y_n<\delta)-\p(\theta_1Y_1+\dots+\theta_{n-1}Y_{n-1}< t)\p(Y_n<\delta)\\
&=\p(\theta_1Y_1+\dots+\theta_{n-1}Y_{n-1}\ge t)\p(Y_n<\delta).
\end{align*}
Therefore, if $Y_1,\dots,Y_n$ are weakly negatively associated, the equal signs in \eqref{eq:proofIN1} and \eqref{eq:proofIN2} will become greater-than-or-equal signs and the proof in step (a) can still go through. Hence, \eqref{eq:maineq1} holds for weakly negatively associated $\mathrm{Pareto}(1)$ random variables.

(c) To show \eqref{eq:maineq1} for WNAID super-Pareto random variables, %. For a super-Pareto random variable $X$,  there exist 
%$Y\sim \mathrm{Pareto}(1)$ and  an increasing, convex, and non-constant function $f$ such that $X\laweq f(Y)$. 
the following two lemmas will be helpful. Their proofs are   in Appendix \ref{app:r3}, and they do not assume super-Pareto distributions. 

\begin{lemma}\label{lem:WNApareto} 
Suppose that $X\laweq f(Y)$ for an increasing    convex  function $f$, $X_1\laweq X$, and 
 $X_1,\dots,X_n$ are  WNAID.    There exist WNAID  random variables $Y_1,\dots,Y_n $ with $Y_1\laweq Y$ such that $$(X_1,\dots,X_n)\laweq(f(Y_1),\dots,f(Y_n)).$$
\end{lemma}
\begin{lemma}
\label{lem:r3-2}
Let $Y,Y_1,\dots,Y_n$ be any random variables, and $(\theta_1,\dots,\theta_n)\in \Delta_n$.
If   $Y  \le_{\rm st}\sum_{i=1}^n\theta_{i} Y_{i}$, 
then   $
f(Y) \le_{\rm st}\sum_{i=1}^n\theta_{i}f(Y_{i})
 $ for any increasing convex function $f$. 
\end{lemma}

By Lemma \ref{lem:WNApareto} and the definition of super-Pareto distribution, for any WNAID super-Pareto random variables $X_1,\dots,X_n$ with $X_1\laweq X$,  there exist WNAID  $Y_{1},\dots,Y_{n}\sim \mathrm{Pareto}(1)$ such that $(X_1,\dots,X_n)\laweq(f(Y_1),\dots,f(Y_n))$. Therefore, for \eqref{eq:maineq1}, it suffices to show 
\begin{equation} \label{eq:f(y)1} 
f(Y) \le_{\rm st}\sum_{i=1}^n\theta_{i}f(Y_{i}),
\end{equation} 
which follows by combining Lemma \ref{lem:r3-2} and the result in step (b).  As it is  shown in the previous steps that  $ Y<_{\rm st}\sum_{i=1}^n\theta_iY_i$, it is clear that  for $t>f(1)$, $\p\(\sum_{i=1}^n\theta_{i}f(X_i)>  t\)> \p\(f(X)> t\)$,  where $\theta_i>0$ for at least two $i\in [n]$.  Hence, the strictness statement after \eqref{eq:maineq1} also holds.

 (d) Next we show (ii).
 Let $A_1,\dots,A_n$ be events  independent of $(X_1,\dots,X_n)$ and event $A$ be independent of $X$ satisfying $\p(A)=\sum_{i=1}^n \theta_i\p(A_i)$.   
 For $S\subseteq [n]$, let $B_S=\(\bigcap_{i\in S }A_i\)\cap\(\bigcap_{i\in S^c }A_i^c\)$. For $(\theta_1,\dots,\theta_n)\in \Delta_n$, we write 
$$\sum_{i=1}^n\theta_iX_i\id_{A_i}=\sum_{S\subseteq[n]}\id_{B_S}\sum_{i\in S}\theta_iX_i.$$
By \eqref{eq:maineq1}, $\sum_{i\in S}\theta_iX_i\ge_{\rm st}\sum_{i\in S}\theta_iX$ for any $S\subseteq[n]$. 
As first-order stochastic dominance is closed under mixture (e.g., Theorem 1.A.3 (d) of \cite{SS07}), $\sum_{i\in S}\theta_iX_i\id_{B_S}\ge_{\rm st}\sum_{i\in S}\theta_iX\id_{B_S}$ for any $S\subseteq[n]$.
% As $A_1,\dots,A_n$ are independent of $(X_1,\dots,X_n)$, by Theorem 1.A.14 of \cite{SS07}, $\sum_{i\in S}\theta_iX_i\id_{B_S}\ge_{\rm st}\sum_{i\in S}\theta_iX\id_{B_S}$ for any $S\subseteq[n]$. \textcolor{red}{Let $Z(\theta)=\theta \id_{B_S}$ where $\theta\in\R$. Since $Z(\theta)\le_{\rm st}Z(\theta')$ when $\theta\le\theta'$, $\{Z(\theta):\theta\in \R\}$ is stochastically increasing in the usual stochastic order. By Theorem 1.A.14 of \cite{SS07}, as $\sum_{i\in S}\theta_iX_i\ge_{\rm st}\sum_{i\in S}\theta_iX$ for any $S\subseteq[n]$ and $A_1,\dots,A_n$ are independent of $(X_1,\dots,X_n)$,  $Z(\sum_{i\in S}\theta_iX_i)\ge_{\rm st}Z(\sum_{i\in S}\theta_iX$) which is $\sum_{i\in S}\theta_iX_i\id_{B_S}\ge_{\rm st}\sum_{i\in S}\theta_iX\id_{B_S}$.}  
Since $B_S$ and $B_R$ are mutually exclusive for any distinct $S,R\subseteq[n]$, we have
\begin{align*}
\sum_{i=1}^n\theta_iX_i\id_{A_i}=\sum_{S\subseteq[n]}\id_{B_S}\sum_{i\in S}\theta_iX_i\ge_{\rm st}\sum_{S\subseteq[n]}\sum_{i\in S}\theta_iX\id_{B_S}.
\end{align*}

Finally, 
we need to show
$\sum_{S\subseteq[n]}(\sum_{i\in S}\theta_i)X\id_{B_S}\ge_{\rm st} X\id_{A}$.
For this, we prove the following statement.  
For  mutually exclusive events  $B_1,\dots,B_n$ 
 independent of $X$ and $(c_1,\dots,c_n)\in[0,1]^n$, we have
\begin{align}\label{eq:R2-1}
X\id_{B}\le_{\rm st}\sum_{i=1}^n c_i X\id_{B_i},\end{align}
where $B$ is an event independent of $X$ satisfying $\p(B)=\sum_{i=1}^nc_i\p(B_i)$. 
To show this, first note that the statement is clearly true if $c_1=\dots=c_n=0$. If any components of $(c_1,\dots,c_n)$ are zero, the problem simply reduces its dimension. Hence, we assume that $(c_1,\dots,c_n)\in(0,1]^n$ for the rest of the proof. 
Let the survival function of $X$ be $\p(X>x)=1/g(x)$ for $x\ge z_X$
and $\p(X>x)=1$ for $x<z_X$. 
As $X$ is regular, $z_X > 0 $, 
%and  $g(0)\ge %0$. For $t\ge 0$ and $c\in(0,1]$, since $g$ is strictly increasing and concave, we have
%$$g\(\frac{t}{c}\)\le g(t)+\(t\frac{1-c}{c}\)\frac{g(t)-g(0)}{t-0}
%    =\frac{1}{c}g(t)-\frac{1-c}{c}g(0)
%    \le \frac{1}{c}g(t).$$
%Hence, 
the concavity of $g$, and $g(x)\le x/z_X$ for $x\ge z_X$ together imply
$g(t/c)\le g(t)/c$ for $t\ge 0$ and $c\in(0,1]$. For $t\ge z_X$, as $B_1,\dots,B_n$ are mutually exclusive and $c_i\in(0,1]$ for all $i\in[n]$,
\begin{align*}
\p\(\sum_{i=1}^n c_i X\id_{B_i}\le t\)
&=1-\sum_{i=1}^n\frac{\p(B_i) }{g(t/c_i)} \le 1-\frac{1}{g(t)}\sum_{i=1}^nc_i\p(B_i) = 1-\frac{\p(B)}{g(t)}=\p(X\id_{B}\le t). %\label{eq:strict}
\end{align*}
For $ t\in [0,z_X)$, $$\p\(\sum_{i=1}^n c_i X\id_{B_i}\le t\)\le\p\(\sum_{i=1}^n c_i X\id_{B_i}\le z_X\)\le 1-\frac{\p(B)}{g(z_X)}=\p(X\id_{B}\le t), $$ where we used $g(z_X)=1$, implied by $1\le g(z_X)\le z_X/z_X$. This yields \eqref{eq:R2-1}.    
As  $\sum_{i\in S}\theta_i\in[0,1]$ for any $S\subseteq[n]$, and 
\begin{align*}
\sum_{S\subseteq[n]}\p(B_S)\sum_{i\in S}\theta_i=\sum_{j=1}^{n}\theta_j\sum_{S\subseteq[n], j\in S}\p( B_{S})=\sum_{j=1}^n\theta_i\p(A_j)= \p(A),
\end{align*}
$\sum_{S\subseteq[n]}(\sum_{i\in S}\theta_i)X\id_{B_S}\ge_{\rm st} X\id_{A}$  follows from \eqref{eq:R2-1} and the desired result  \eqref{eq:maineq2} holds.

% In this section, we present the proof of Theorem \ref{thm:1} which is the most important result of the paper. Proofs of all other results are in Section \ref{app:proof}. 
%The following two lemmas are used in the proof of Theorem \ref{thm:1}.
%\begin{proof}[Proof of Theorem \ref{thm:1}]

\subsection{A few remarks}

The next few remarks discuss the relation of Theorem \ref{thm:1} to the literature and some technical issues.
%More general models for which diversification penalty 
%\eqref{eq:maineq1} holds are discussed in Section \ref{sec:extension}.

\begin{remark}[EVT]
In the literature of EVT, it has been observed that, for iid extremely heavy-tailed Pareto losses $X_1,\dots,X_n$,  
$$ \p\(\frac 1n\sum_{i=1}^n X_{i} >  t\)   \ge  \p\(X> t\)  $$  holds true asymptotically as $t\to \infty$; see, e.g.,  \cite{KGT04}, \cite{AAK06}, and \cite{ELW09}.
%\footnote{Such an asymptotic result also holds for dependent extremely heavy-tailed Pareto losses (e.g., \cite{ELW09} and \cite{kley2016risk}).}
 Theorem \ref{thm:1} implies that the same inequality holds for any $t\in \R$ regardless of whether $t$ is large enough. This gives rise to implications for decision makers whose preferences are not determined purely by the tail behaviour of risks. %; see Section \ref{sec:4}. %and  \ref{sec:5}.
\end{remark}

\begin{remark}[Stable distributions]
As an important class of heavy-tailed distributions, stable distributions have frequently appeared in the analysis of portfolio diversification (e.g., \cite{ibragimov2005new,IW07,ibragimov2008portfolio}). Using majorization order, \cite{ibragimov2005new} showed that diversification increases the risk of a portfolio which consists  of iid stable random variables without finite mean. In  particular, if the stable random variables are one-sided on the positive axis, diversification will increase the total loss in  first-order stochastic dominance; \cite{ibragimov2010optimal} applied  this result to study the problem of optimal bundling strategies for extremely heavy-tailed valuations.
On the other hand,  if the stable random variables are not one-sided, diversification will make the total loss ``more spread out", hence different from  first-order stochastic dominance. These results were extended to the case when losses are convex transformations of iid infinite-mean stable random variables  in \cite{ibragimov2008portfolio}. For iid symmetric infinite-mean stable random variables truncated by a sufficiently large number, diversification still makes the total loss ``more spread out", as shown by \cite{IW07}.  %(Theorem \ref{prop:bounded} below).  
\end{remark}

\begin{remark}[Notions of negative dependence] 
\label{rem:ND}
Among the following notions of negative dependence, weak negative association is weaker than (a) and (b) below, and stronger than (c).
\begin{enumerate}[(a)]
\item  A random vector $\mathbf X= (X_1,\dots,X_n)$ is \emph{negatively associated}   if for every pair of disjoint sets $A$, $B$ of $[n]$ and any functions  $f$ and $g$ both increasing or decreasing coordinate-wise, provided the covariance below exists,
\begin{equation}\label{eq:NA}
\cov(f(\mathbf X_A),g(\mathbf X_B))\le 0,
\end{equation}
where $\mathbf X_A=(X_k)_{k\in A}$ and $\mathbf X_B=(X_k)_{k\in B}$ (\cite{AS81}  and \cite{JP83}). %We say $Z_1,\dots,Z_n$ are negatively associated if $(Z_1,\dots,Z_n)$ is negatively associated.
It is known that random vectors following multivariate normal distributions with non-positive correlations are negatively associated, and so are those obtained from increasing transforms of such normal random vectors  (\cite{JP83}). Choosing $A=  \{i\}$, $B=[n]\setminus A$, $f(y)=\id_{\{y\le x\}}$ and  $g(\mathbf y)=\id_{\{\mathbf y\in S\}}$ in \eqref{eq:NA}  yields \eqref{eq:WNA}, and hence weak negative association is implied. 
\item 
% A random vector $\mathbf X=(X_1,\dots,X_n)$ is said to be \emph{stochastically decreasing in $Y$} if $\mathbb E[g(\mathbf X)\mid Y=y]$ is decreasing in $y$ whenever $g$ is a coordinatewise increasing function such that the conditional expectation exists. 
A random vector  $\mathbf X$ is \emph{negative regression dependent} if for every $i\in[n]$,  the random variable $\E[g(\mathbf X_{-i})| X_i]$
is a decreasing function of $X_i$ for any   coordinate-wise increasing function $g$ such that the conditional expectation exists; see \cite{L66}, who only formulated the case $n=2$. This notion for general $n>2$ is called negative dependence through stochastic ordering by \cite{block1985concept}. 
To check that this notion is stronger than weak negative association, it suffices to take the function $g(\mathbf x)=-\id_{\{\mathbf x\in S\}}$ for a decreasing set $S\subseteq \R^{n-1}$.
 
 \item
 A random vector $\mathbf X=(X_1,\dots,X_n)$ is \emph{negatively orthant dependent}  if for all $\mathbf x =(x_1,\dots,x_n) \in\R^n$,  $\p(\mathbf X\le \mathbf x)\le \prod_{i=1}^n\p(X_i\le x_i)$ and $\p(\mathbf X >  \mathbf x)\le \prod_{i=1}^n\p(X_i > x_i)$.
The fact that \eqref{eq:WNA} implies negative orthant dependence follows from \eqref{eq:WNA} and \eqref{eq:inset}. 
Negative orthant dependence is not sufficient for the proof of Theorem \ref{thm:1}, because in the  proof (see Section \ref{sec:proof}, step (b)) we need the inequality 
$\p(\theta_1 Y_1+\dots+\theta_{n-1}Y_{n-1}<t, Y_n\le \delta)  \le \p(\theta_1 Y_1+\dots+\theta_{n-1}Y_{n-1}<t)\p( Y_n\le \delta) $, 
which holds under weak negative association of $(Y_1,\dots,Y_n)$ but not under negative orthant dependence. 

\end{enumerate}
\end{remark}

\begin{remark}[Majorization]
\label{rem:open}
Let $X_{1},\dots,X_{n}$ be  iid super-Pareto random variables.
Inspired by Theorem \ref{thm:1},
a question is 
 whether  
% $ 
% \sum_{i=1}^kX_i /k \le_{\rm st} \sum_{i=1}^\ell X_i /\ell,
% $ holds for all $k,\ell \in \mathbb N$ such that $k\le \ell$. Theorem \ref{thm:1} implies that this statement is true if $\ell$ is a multiple of $k$.
%A stronger conjecture is 
%whether  
\begin{align}\label{eq:conj}
 \sum_{i=1}^n\eta_i X_{i}\le_{\rm st} \sum_{i=1}^n\theta_iX_{i}
\end{align}
 holds for two vectors $(\theta_{1},\dots,\theta_n)\in \Delta_n$ and $(\eta_{1},\dots,\eta_n) \in \Delta_n$
increasing in majorization order; that is, $\sum_{i=1}^n \phi(\theta_i) \le \sum_{i=1}^n \phi(\eta_i)$ for all continuous and convex functions $\phi$ (see \cite{MOA11}). Theorem \ref{thm:1} corresponds to the case   $(\eta_{1},\dots,\eta_n)=(1,0,\dots,0)$. 
This question seems to be beyond the current techniques. 
For results similar to \eqref{eq:conj} on some   distributions, see \cite{proschan1965peakedness}
 and \cite{ibragimov2005new}.

\end{remark}

\section{Conclusion}\label{sec:7}

Our main result (Theorem \ref{thm:1}) establishes   that a weighted average of  WNAID super-Pareto random variables, possibly triggered by  different events, is larger than one such loss  in the  sense of  first-order stochastic dominance. %The class of super-Pareto distributions includes many distributions that are more heavy-tailed than Pareto$(1)$ distribution, such as extremely heavy-tailed Pareto distributions.
This result implies that the diversification of many super-Pareto losses without finite mean increases the risk assessment of a portfolio, according to the superadditivity of VaR.
%for all preferences satisfying (i) choice under risk and (ii) less loss is better. 
Some technical questions remain open and are discussed in Remark \ref{rem:open}. % and Section \ref{sec:4}.

\begin{appendices}

\section{Proofs of all other results}
\label{app:r3}

\begin{proof}{Proof of Proposition \ref{prop:R2-1}}
We first show the   equivalence statement. 
For the ``$\Leftarrow$" direction, let $\p(X\le x)=1-1/g(x)$ for $x\in[z_X,\infty)$, where
 $g
 $ is strictly increasing and concave 
 on $[z_X,\infty)$.
Let 
$f(y)=g^{-1}(y)$ for $y>g(z_X)$ and $f(y)=z_X$ for $1\le y \le g(z_X)$. 
It is straightforward to see that for any  Pareto$(1)$ random variable $Y$, $f(Y)\laweq X$. Next, we show the ``$\Rightarrow$" direction. For $x<\infty$,  the right-continuous generalized inverse of $f$ is
$f^{-1+}(x)=\inf \{t:f(t)>x\}.$
For $x\ge f(1)$, $\p(f(Y)\le x)=\p(Y\le f^{-1+}(x))=1-1/f^{-1+}(x)$. As $f$ is increasing, convex, and non-constant, $g:=f^{-1+}$ is strictly increasing and concave. Hence, $g$ is concave and strictly increasing on $[z_X,\infty)$. %The statement on infinite mean follows because a strictly increasing convex function $f$ satisfies, for some $a>0$ and $b\in \R$, $f(x)\ge ax +b$ for all $x\in \R$.

The  statement on regularity follows by observing two facts.
First, $f(1)>0$ implies that $z_X=f(1)>0$. 
Moreover, since $f$ is convex and $f(0)=0$,  we have 
$f(y) \ge y f(1)$ for all $y>1$, which gives $g(x)\le x/z_X$ for $x\ge z_X$ via $g=f^{-1+}$.
%The condition $z_X>0$ is equivalent to $f(1)>0$. The condition $g(x)\le x/z_X$ for $x\ge z_X$ is equivalent to $f(y)\ge y(f(1)-f(0))$.
\end{proof}

\begin{proof}{Proof of Proposition \ref{prop:equality}}
Note that \eqref{eq:*} implies that 
$\ES_p(X)\le \ES_p(\sum_{i=1}^n \theta_i X_i)$
for all $p\in (0,1)$, where $\ES_p$ is defined  as 
$$\ES_{p}(X)=\frac{1}{1-p}\int_{p}^{1}\VaR_{u}(X)\mathrm{d}u.$$  Since $\ES_p$ is convex and $X_1,\dots,X_n$ are identically distributed, we have 
$$\ES_p(X)\le \ES_p\left (\sum_{i=1}^n  \theta_i   X_i\right) \le\theta_i  \sum_{i=1}^n\ES_p(X_i) = \ES_p(X),~~p\in (0,1). $$
Using  positive homogeneity of $\ES_p$, it follows that the equality $ \sum_{i=1}^n\ES_p(\theta_i  X_i)= \ES_p(\sum_{i=1}^n \theta_i X_i)$ holds for each $p\in(0,1)$. By Theorem 5 of \cite{WZ21}, this implies that $(\theta_1 X_1,\dots,\theta_n X_n)$ is $p$-concentrated  for each $p$; this result requires $X_1,\dots,X_n$ to have finite mean. Using Theorem 3 of \cite{WZ21},  the above condition implies
  that $(X_1,\dots,X_n)$ is comonotonic. For definitions of comonotonicity and $p$-concentration, see \cite{WZ21}. 
Since $X_1,\dots,X_n$ are identically distributed, comonotonicity further implies that $X_1=\dots=X_n$ almost surely. 
\end{proof}

\begin{proof}{Proof of Lemma \ref{lem:WNApareto}}  
If $f$ is constant, then there is nothing to show. 
If $f$ is not constant, then there exists $z\in \R$ such that $f(x)$ is strictly increasing for $x>z$.
%Note that $F(x)$ is strictly increasing  for $x\ge z_X$.  
Denote by $q=\p(X=z_X )$. 
For $i\in [n]$, let $W_i$ be a uniform transform of $X_i$, that is, $W_i$ is a standard uniform random variable such that $F^{-1}_{X} (W_i)=X_i$ a.s.
A uniform transform always exists in an atomless probability space; see Lemma A.32 of \cite{FS16}. 
For all $i\in[n]$, let $A_i=\{X_i\le z_X\}$ and 
$$U_i=qV_i\id_{A_i}+W_i\id_{A_i^c},$$
where $V_1,\dots,V_n$ are iid standard uniform random variables   independent of $(X_1,\dots,X_n)$.
%Random variables $U_1,\dots,U_n$ are called distributional transforms of $X_1,\dots,X_n$ and $X_i=F_X^{-1}(U_i)$ almost surely for all $i\in[n]$ (see, e.g., Proposition A.6 in \cite{MFE15}). 
Note that for each $i\in [n]$, $U_i$ is also a uniform transform of $X_i$, and 
the distribution of $(U_1,\dots,U_n)$ is one possible copula of $(X_1,\dots,X_n)$.
Let $Y_i=F^{-1}_Y(U_i)$ for $i\in [n]$. 
Note that $F^{-1}_X=f\circ F^{-1}_Y$ and $f$ is strictly increasing for $x>z$. 
Hence $Y_i$ is a strictly increasing function of $X_i$ given $V_i$ for each $i\in [n]$. Moreover, $(X_1,\dots,X_n)\laweq(f(Y_1),\dots,f(Y_n))$ as they have the same copula and marginal distributions.
It remains to show that $Y_1,\dots,Y_n$ are weakly negatively associated. For any decreasing set $A\subseteq\R^{n-1}$ and $x\in \R$ with $F_Y(x)>0$,   denote by $\beta =1/F_Y(x)$, and we have 
\begin{align*}
\p((Y_1,\dots,Y_{n-1})\in A|Y_n\le x)&=\p((Y_1,\dots,Y_{n-1})\in A,Y_n\le x)\beta \\
&=\E[\p((Y_1,\dots,Y_{n-1})\in A,Y_n\le x|V_1,\dots,V_n)]\beta\\
&\le \E[\p((Y_1,\dots,Y_{n-1})\in A|V_1,\dots,V_n)\p(Y_n\le x|V_1,\dots,V_n)] \beta \\
&=\E[\p((Y_1,\dots,Y_{n-1})\in A|V_1,\dots,V_{n-1})\p(Y_n\le x|V_n)] \beta \\
&=\E[\p((Y_1,\dots,Y_{n-1})\in A|V_1,\dots,V_{n-1})]\E[\p(Y_n\le x|V_n)] \beta \\
&=\p((Y_1,\dots,Y_{n-1})\in A)\p(Y_n\le x) \beta =\p((Y_1,\dots,Y_{n-1})\in A),
\end{align*}
where the inequality holds because  for each $i\in [n]$, conditional on $V_i$, $Y_i$ is a strictly increasing function of $X_i$, and  $X_1,\dots,X_n$ are weakly negatively associated.
Therefore,  $Y_1 ,\dots, Y_n$ are also weakly negatively associated. 
% is strictly increasing and $X_1,\dots,X_n$ are weakly negatively associated. 
% Hence, $U_1,\dots,U_n$ are weakly negatively associated.  By  Proposition \ref{prop:R2-1},  there exists 
% $Y\sim \mathrm{Pareto}(1)$ and  an increasing, convex, and non-constant function $f$ such that $X\laweq f(Y)$. Let $(Y_1,\dots,Y_n)=(F_Y^{-1}(U_1),\dots,F_Y^{-1}(U_n))$, which clearly has $\mathrm{Pareto}(1)$ marginals. Moreover, $Y_1,\dots,Y_n $ are  weak negatively associated because  weak negative association is preserved under marginal increasing transforms by definition. Since $f(F_Y^{-1}(t))=F_{X}^{-1}(t)$,  we obtain   $(X_1,\dots,X_n)\laweq(f(Y_1),\dots,f(Y_n))$, the desired statement. 
 \end{proof}

\begin{proof}{Proof of Lemma \ref{lem:r3-2}}
  As $f$ is convex and increasing, we have $f( Y)\le_{\rm st}f(\sum_{i=1}^n\theta_iY_i)\le \sum_{i=1}^n\theta_{i}f(Y_{i}),$  where the first inequality holds as \eqref{eq:maineq1} holds for $Y_{1},\dots,Y_{n}$ and the second inequality is due to the convexity of $f$. 
\end{proof}

\end{appendices}

% Acknowledgments here
\ACKNOWLEDGMENT{The authors thank an Area Editor, an Associate Editor, and three anonymous referees for various relevant comments on the paper. In particular, the inclusion of negative dependence was inspired by a comment from an Associate Editor, and the open question in Remark \ref{rem:open} was inspired by comments from the referees. 
The authors thank Yuming Wang and Wenhao Zhu for  ideas that led to a preliminary version of  Theorem \ref{thm:1}. 
We also thank  Hansj\"org Albrecher, Jan Dhaene, John Ery, Taizhong Hu, Massimo Marinacci,  Alexander Schied, and Qihe Tang for helpful comments on a previous version of the paper.
Ruodu Wang is supported by the Natural Sciences and Engineering Research Council of Canada (RGPIN-2018-03823 and CRC-2022-00141).}

%\theendnotes
 
{

}
\end{document}